\documentclass[12pt]{article}
\author{D.S. Malyshev\footnote{National Research University Higher School of Economics, 25/12 Bolshaja Pecherskaja Ulitsa, Nizhny Novgorod, 603155, Russia;
Lobachevsky State University of Nizhny Novgorod, 23 Gagarina Avenue, Nizhny Novgorod, 603950, Russia; Email: dsmalyshev@rambler.ru}$~$ and O.O. Lobanova\footnote{ Lobachevsky State University of Nizhny Novgorod, 23 Gagarina Avenue, Nizhny Novgorod, 603950, Russia; Email: olga-olegov@yandex.ru}}
\date{}
\title{The coloring problem for $\{P_5,\overline{P_5}\}$-free graphs and $\{P_5,K_p-e\}$-free graphs is polynomial}

\usepackage{amsmath,amsfonts,amssymb}
\usepackage{latexsym}

\begin{document}
\maketitle
\newtheorem{theorem}{Theorem}
\newtheorem{corollary}{Corollary}
\newtheorem{conjecture}{Conjecture}
\newtheorem{lemma}{Lemma}
\newenvironment{proof}[1][Proof]{\textbf{#1.} }{\ \rule{0.5em}{0.5em}}

\begin{abstract}

We show that determining the chromatic number of a $\{P_5,\overline{P_5}\}$-free graph or a
$\{P_5,K_p-e\}$-free graph can be done in polynomial time.\\

Keywords: computational complexity, coloring problem, hereditary class, efficient algorithm
\end{abstract}

\section{Introduction}

A \emph{coloring} is an arbitrary mapping from the set of vertices or
edges of a graph into a set of colors of the graph such that any adjacent
vertices (or edges) are colored with different colors. The minimal number of colors
sufficient for coloring a graph $G$ is said to be the \emph{chromatic number} of $G$
denoted by $\chi(G)$. The \emph{coloring problem} is to decide whether $\chi(G)\leq k$ or not
for given graph $G$ and a number $k$. A similar \emph{$k$-colorability problem} is to check
whether a given graph can be colored with at most $k$ colors. Both problems can be naturally defined
in another way via partition into independent sets.
An \emph{independent set} of graph is an arbitrary set of pairwise nonadjacent vertices. A coloring
is partitioning of vertex set of a graph into independent subsets called \emph{color classes}.

There is a natural lower bound for the chromatic number of a graph. A \emph{clique} in a graph
is a subset of pairwise adjacent vertices. The size of a maximum clique in a graph $G$ is called the
\emph{clique number} of $G$ denoted by $\omega(G)$. Clearly, $\chi(G)\geq \omega(G)$. Sometimes,
computing $\omega(G)$ helps to determine $\chi(G)$ \cite{GLS84,M15}.

A class of graphs is called \emph{hereditary} if it is closed under isomorphism and
deletion of vertices. It is well known that any hereditary (and only
hereditary) graph class ${\mathcal X}$ can be defined by a set of
its forbidden induced subgraphs ${\mathcal S}$. We write ${\mathcal
X}=Free({\mathcal S})$ in this case, and graphs in ${\mathcal X}$ are said to
be \emph{${\mathcal S}$-free}. If ${\mathcal S}=\{G\}$, then we write ''$G$-free'' instead of ''$\{G\}$-free''.

We say that ${\mathcal X}$ is \emph{easy} for the coloring problem if ${\mathcal X}$ is hereditary and the problem can be
polynomially solved for it.

The computational complexity of the coloring problem was completely determined for
all classes of the form $Free(\{G\})$ \cite{KKTW01}. A study of forbidden
pairs was also initiated in \cite{KKTW01}. A complete complexity dichotomy appeared
hard to obtain even in the cases of two four-vertex and connected five-vertex forbidden
induced subgraphs \cite{LM15,M14}. For all but three cases either NP-completeness or polynomial-time
solvability was shown in the family of hereditary classes defined by four-vertex forbidden induced structures \cite{LM15}.
The remaining three classes are stubborn. A similar result was obtained in \cite{M14} for two
connected five-vertex forbidden induced fragments, where the number of open cases was 13. Recently,
it was reduced to 11 \cite{M15}. We reduce the number to nine by showing that the coloring problem can
be solved for $\{P_5,\overline{P_5}\}$-free and $\{P_5,K_p-e\}$-free graphs in polynomial time.

\section{Notation}

As usual, $P_n,C_n,O_n,K_n$ stand respectively for the simple path,
the chordless cycle, the empty graph, the complete graph with $n$ vertices respectively. A graph $K_p-e$
is obtained from $K_p$ by deleting an arbitrary edge. A formula $N(x)$ means the neighborhood of a vertex
$x$ of some graph. For a graph $G$ and a set $V'\subseteq V(G)$, $G(V')$ denotes the subgraph of $G$ induced by $V'$.

We refer to textbooks in graph theory for any graph
terminology undefined here.

\section{Auxilliary results}

\subsection{Decomposition by clique separators and its applications to the coloring problem}

A \emph{clique separator} in a graph is a clique whose removal increases the
number of connected components. For example, the graph $K_p-e$ has a clique separator
with $p-2$ vertices. If a graph $G$ has a clique separator $Q$, then $V(G)\setminus Q$ can
be arbitrarily partitioned into nonempty subsets $A$ and $B$ such that no vertex of $A$
is adjacent to a vertex of $B$. Let $G_1\triangleq G(A\cup Q)$ and $G_2\triangleq G(B\cup Q)$.
We repeat a similar decomposition until no further decomposition is possible. The whole process
can be represented by a binary decomposition tree whose leaves correspond to some induced subgraphs of $G$
without clique separators. There exists an $O(mn)$-time algorithm for constructing some binary decomposition
tree for any graph with $n$ vertices and $m$ edges \cite{T85}.

\begin{lemma}
For each graph $G$, $\chi(G)=\max(\chi(G_1),\chi(G_2))$.
\end{lemma}

\begin{proof}
Without loss of generality, $\chi(G_1)\leq \chi(G_2)$. Let us consider a partial coloring
of $G$ induced by an optimal coloring of $G_2$ and color classes of $G_1$ in its optimal coloring containing all vertices of $Q$.
These color classes can be colored with colors assigned to elements of $Q$ in the partial coloring. To color the remaining part of $A$, it
is enough $\chi(G_1)-|Q|$ colors distinct to the colors of $Q$. The set $B$ has $\chi(G_2)-|Q|\geq \chi(G_1)-|Q|$
color classes with colors of this type. Hence, $G$ can be colored with $\chi(G_2)$ colors. So, $\chi(G)=\chi(G_2)$.
\end{proof}

A maximal induced subgraph of a given graph without proper clique separators will be called a \emph{$C$-block} of the graph.  Leaves of a decomposition tree of any graph correspond to its $C$-blocks. Let ${\mathcal X}$ be a class of
graphs. The set of all graphs whose every $C$-block belongs to ${\mathcal X}$ will be called the
\emph{$C$-closure of ${\mathcal X}$} denoted by $[{\mathcal X}]_C$.

\begin{theorem}
If ${\mathcal X}$ is easy for the coloring problem, then it is so for $[{\mathcal X}]_C$.
\end{theorem}

\begin{proof} Clearly, $[{\mathcal X}]_C$ is hereditary. All $C$-blocks of a graph $G\in [{\mathcal X}]_C$
belong to ${\mathcal X}$, and the coloring problem can be solved in polynomial time for them.
A decomposition tree for $G$ can be constructed in polynomial time. Hence, by the previous lemma, $[{\mathcal X}]_C$ is
easy for the coloring problem.
\end{proof}

\subsection{Modular decomposition and its applications to the weighted coloring problem}

A set $M\subseteq V(G)$ is a \emph{module} in a graph $G$ if either $x$ is adjacent to all elements of $M$ or
none of them for each $x\in V(G)\setminus M$. Each vertex of $G$ and the set $V(G)$ constitute a module called \emph{trivial}.
A module $M$ is a \emph{nontrivial module} in $G$ if $|M|>1$ and $M\neq V(G)$. A graph containing no nontrivial modules is said to
be \emph{prime}. For instance, $P_4$ is prime and $C_4$ does not.

Modular decomposition of graphs is an algorithmic technique based on the following decomposition theorem due to T. Gallai.

\begin{theorem} \cite{G67}
Let $G$ be a graph with at least two vertices. Then exactly one of the following conditions holds:

$(1).$ $G$ is not connected

$(2).$ $\overline{G}$ is not connected

$(3).$ $G$ and $\overline{G}$ are connected, and there is a set $V'$ with at least four elements and an unique partition $P(G)$ of $V(G)$
such that

$(a).$ $G(V')$ is a maximal prime induced subgraph of $G$

$(b).$ for each $V''\in P(G)$, $V''$ is a module (perhaps, trivial) in $G$ and $|V''\cap V'|=1$.
\end{theorem}

By the theorem, there are decomposition operations of three types. First, if $G$ is not connected, then
disconnect it into connected components $G_1,\ldots,G_p$. Second, if $\overline{G}$ has connected components $\overline{G_1},\ldots,\overline{G_q}$,
then decompose $G$ into $G_1,\ldots,G_q$. At length, if $G$ and $\overline{G}$ are connected,
then its maximal modules are pairwise disjoint, and they form the partition $P(G)$. The graph $G$ is decomposed into subgraphs in
$\{G(V'')|~V''\in P(G)\}$. Additionally, each class of $P(G)$ is contracted
to obtain a graph which is isomorphic to $G(V')$. In other words, $G(V')$
is an induced subgraph of $G$ producing by taking one element in each class of $P(G)$.

The decomposition process above can be represented by an uniquely determined tree called
the \emph{modular decomposition tree} of $G$. Its vertices are induced subgraphs of $G$.
A vertex $G$ has the connected components of $G$ or $\overline{G}$ as the children in the
first two cases; the children are subgraphs of the form $G(V''), V''\in P(G)$ in the third one.
Moreover, we associate the graph $G(V'')$ with the vertex $G$. The modular decomposition tree
can be determined in $O(n+m)$-time for any graph with $n$ vertices and $m$ edges \cite{CH05}.

The \emph{weighted coloring problem} is to find, for given $G$ and a function $w: V(G)\rightarrow {\mathbb N}$, the smallest number $k$
such that there is a function $c: V(G)\rightarrow 2^{\{1,2,\ldots,k\}}$ such that $|c(v)|=w(v)$ for any $v$ and $c(v_1)\cap c(v_2)=\emptyset$ for
any adjacent $v_1$ and $v_2$. The elements of $c(v)$ are called the \emph{colors of} $v$. This $k$ is denoted by $\chi_w(G)$ and called the \emph{weighted chromatic number} of $G$. For every graph $G$, $\chi_{w'}(G)=\chi(G)$, where $w'$ maps every vertex to 1.

Clearly, for each function $w$, we have $\chi_w(G)=\max\limits_{i}(\chi_w(G_i))$, where $G_1,\ldots,G_p$ are connected components of $G$.
Similarly, if $\overline{G_1},\ldots,\overline{G_q}$ are connected components of $\overline{G}$, then $\chi_w(G)=\sum\limits_{i=1}^{q}\chi_w(G_i)$.

\begin{lemma} Let $G$ be a graph, $P(G)$ be its modular decomposition, $w: V(G)\rightarrow {\mathbb N}$ be an arbitrary function. Then $\chi_w(G)=\chi_{w^{*}}(G(V'))$,
where $w^{*}(v)=\chi_w(G(V''))$ for each $v\in V', V''\in P(G), \{v\}=V'\cap V''$.
\end{lemma}

\begin{proof}
Contraction of $V''$ to $v$ and assignment $w(v)=\chi_w(G(V''))$ produces a subgraph whose weighted
chromatic number is at most $\chi_w(G)$. On the other hand, each element of $N(v)$ cannot have some $\chi_w(G(V''))$
colors of $v$. Hence, the weighted chromatic number of the subgraph is equal to $\chi_w(G)$. Therefore, $\chi_w(G)=\chi_{w^{*}}(G(V'))$.
\end{proof}

Let $[{\mathcal X}]_P$ be the set of graphs whose every prime induced subgraph
belongs to ${\mathcal X}$. Clearly, $[{\mathcal X}]_P$ is hereditary whenever ${\mathcal X}$ is hereditary.
The theorem below follows from the previous lemma and \cite{CH05}.

\begin{theorem}
If ${\mathcal X}$ is an easy class for the coloring problem, then it is so for $[{\mathcal X}]_P$.
\end{theorem}

\subsection{Bipartite Ramsey theorem}

A famous Ramsey theorem claims that any graph has a sufficiently large independent set or a
sufficiently large clique. There are numerous its analogues for different classes of graphs, e.g. for
bipartite graphs. Recall that a graph is \emph{bipartite} if its vertex set can be partitioned into at most two independent sets.
These independent sets are called \emph{parts}. A \emph{matching} in a graph is a subset of pairwise nonadjacent edges.
The following result is a corollary of theorem 2 from \cite{EHP00} for $H=K_{s,s}$.

\begin{lemma}
Any bipartite graph $G$ having parts $A$ and $B$ with $n>s^{s+1}$ vertices contains subsets $A'\subseteq A, B'\subseteq B, |A'|=|B'|=\lfloor(\frac{n}{s})^{\frac{1}{s}}\rfloor$ such that $G(A'\cup B')$ is empty or complete bipartite.
\end{lemma}

\subsection{Connected $\{P_5,K_p-e\}$-free graphs without clique separators}

Let $G$ be a connected $\{P_5,K_p-e\}$-free graph ($p\geq 3$) without clique separators, and let $Q$ be its
maximum clique.

\begin{lemma}
The graph $G$ is $O_3$-free or $|Q|\leq (p+1)^{p+2}(p-2)$.
\end{lemma}

\begin{proof}
Assume that $|Q|>(p+1)^{p+2}(p-2)$. Let $N(Q)\triangleq \{y\not \in Q|~ \exists x\in Q, (y,x)\in E(G)\}$. Any element of $N(Q)$ cannot be adjacent to
$p-2$ or more vertices of $Q$. Let us consider a bipartite graph $G'$ induced by edges between $Q$ and $N(Q)$. As $G$ has
no clique separators, $Q$ and $N(Q)$ are parts of $G'$. Clearly, the graph $G'$ has a matching with $\lfloor\frac{|Q|}{p-2}\rfloor$ edges,
and it is $K_{p-2,p-2}$-free. Let $N_1\triangleq \{u_1,u_2,\ldots,u_k\}$ be a maximum subset of $Q$ such that $N(Q)$ has vertices
$v_1,v_2,\ldots,v_k$ with $v_i\in N(u_i)\setminus \bigcup\limits_{j\neq i}N(u_j)$ for each $i$. By the previous lemma for $s=p+1$,
$k\geq \lfloor\frac{|Q|}{p-2}\rfloor\geq p+1$. As $p\geq 3$, $N_2\triangleq\{v_1,v_2,\ldots,v_k\}$
must be an independent set or a clique to avoid an induced $P_5$. If $N_2$ is independent, then there is no a vertex $v_i$ having a neighbor
$w\not \in Q\cup N(Q)$. Otherwise, $w$ must be adjacent to all vertices of $N_2$, and $G$ is not $P_5$-free. Hence, a possible
neighbor $w\not \in Q$ of an element $v_i\in N_2$ must belong to $N(Q)$. To avoid an induced $P_5$, $w$ must be adjacent to all elements of $N_2$ or to
$v_i$ only. The second case is realized if and only if $N_1$ has only one neighbor of $w$ coinciding with $u_i$. In the first case,
there are some three non-neighbors $u_{i_1},u_{i_2},u_{i_3}$ of $w$, as $G'$ is $K_{p-2,p-2}$-free. But $v_{i_1},w,v_{i_2},u_{i_2},u_{i_3}$ induce $P_5$. Hence, any possible neighbor $w_i$ of $v_i$ that lies outside $Q$ must be adjacent to $u_i$ and nonadjacent to $u_1,\ldots,u_{i-1},u_{i+1},\ldots,u_k$. Similarly, $N(w_i)\subseteq N(u_i)\cup \{u_i\}$. Hence, $Q$ is a clique separator. Thus, $N_2$ must be a clique.

Let $Q'$ be a maximal clique that includes $N_2$. Suppose that
$v\in N(Q)\setminus Q'$. Since $N_1$ is maximum, $v$ has neighbors in $N_1$, say, $u_1,\ldots,u_q$. Clearly, $q\leq p-3$. To avoid an induced $P_5$, $v$ must
be adjacent to at least $k-q-1$ vertices among $v_{q+1},\ldots,v_k$. Similarly, $v$ must be adjacent to $v_1,\ldots,v_q$. Hence, $v$
is adjacent to at least $k-1$ vertices of $N_2$. To avoid an induced $K_p-e$, $v\in Q'$. Thus, $N(Q)\setminus Q'=\emptyset$.
In fact, $Q'=N(Q)$ and $V(G)=Q\cup N(Q)$, since $N(Q)$ is a clique separator otherwise. So, $G$ is $O_3$-free.\end{proof}

\subsection{Connected prime $\{P_5,\overline{P_5}\}$-free graphs}

A graph is said to be \emph{perfect} if the clique number and the chromatic number are equal
for every its induced subgraph (not necessarily proper). The class of perfect graphs coincides with $Free(\{C_5,\overline{C_5},C_7,\overline{C_7},\ldots\})$,
by the strong perfect graph theorem \cite{CRST06}.

\begin{lemma}
Any connected prime $\{P_5,\overline{P_5}\}$-free graph is perfect or isomorphic to $C_5$.
\end{lemma}

\begin{proof}
Every $\{P_5,\overline{P_5},C_5\}$-free graph is perfect, by the strong perfect graph theorem.
Let $G$ be a connected prime $\{P_5,\overline{P_5}\}$-graph containing an induced $C_5$. Every element of $V(G)\setminus V(C_5)$ is either adjacent to all vertices of $C_5$ or to none of them or to two nonadjacent or to three consecutive \cite{F93}. Let $V_i$ be the set of vertices of $G$ adjacent to the $(i-1)$-th and $(i+1)$-th
vertices of $C_5$ counting modulo 5. Let $V_0$ be the set of vertices adjacent to all vertices of $C_5$.
Any element of $V_i$ is adjacent to each element of $V_0\cup V_{i-1}\cup V_{i+1}$, nonadjacent to any element of $V_{i+2}\cup V_{i-2}$, any element of $V_i\setminus V(C_5)$ cannot have neighbors outside $\bigcup\limits_{i=0}^{5}V_i\cup V(C_5)$ \cite{F93}. Suppose that $G$ is not isomorphic to $C_5$. Then $V_i$ has at least two elements for some $i$ or $V_0\neq \emptyset$ and $|V_1|=|V_2|=|V_3|=|V_4|=|V_5|=1$. The set $V_i$ is a nontrivial module in the first case, and $V(C_5)$ is a nontrivial module in the second one. We have a contradiction with the assumption.
\end{proof}

\section{Main result}

\begin{theorem}
The class $Free(\{P_5,\overline{P_5}\})$ and all classes of the form $Free(\{P_5,K_p-e\})$ are easy for the coloring problem.
\end{theorem}

\begin{proof}
It is known that for any $P_5$-free graph $G$ the inequality $\chi(G)\leq 4^{w(G)-1}$ holds \cite{G87}. Moreover, for each fixed $k$,
the $k$-colorability problem can solved in polynomial time for $P_5$-free graphs \cite{HKLSS10}. Hence, by these results, Theorem 1 and Lemma 4,
the coloring problem for $\{P_5,K_p-e\}$-free graphs can be polynomially reduced to the same problem for $O_3$-graphs.
The coloring problem for $O_3$-free graphs is polynomially equivalent to determining the sizes of maximum matchings in the complement graphs. The last problem is known to be polynomial \cite{E65}. Hence, $\{P_5,K_p-e\}$-free graphs constitute an easy class for the coloring problem. The class of perfect graphs is easy for the weighted coloring problem \cite{GLS84}. Perfect graphs can be recognized in polynomial time \cite{CCLSV05}. Hence, by these facts, Theorem 3 and Lemma 5,
$Free(\{P_5,\overline{P_5}\})$ is easy for the coloring problem.
\end{proof}

\end{document}